\documentclass[12pt]{article}
\usepackage{amssymb,amsthm,amsfonts,cancel}
\usepackage[sumlimits,intlimits,namelimits]{amsmath}
\usepackage{makeidx}
\usepackage{geometry}
\usepackage[english]{babel}
\usepackage{graphicx}
\usepackage{multirow}
\usepackage{enumitem}
\usepackage{tikz-cd}
\usepackage{url}
\usepackage{quiver}
\usepackage{authblk}
\usepackage{color}
\usepackage{natbib}
\usepackage{todonotes}
\usepackage{notoccite}
\usepackage{subcaption}
\usepackage{float}

\usetikzlibrary{calc}
\usetikzlibrary{decorations.pathmorphing}

\tikzset{curve/.style={settings={#1},to path={(\tikztostart)
    .. controls ($(\tikztostart)!\pv{pos}!(\tikztotarget)!\pv{height}!270:(\mtikztotarget)$)
    and ($(\tikztostart)!1-\pv{pos}!(\tikztotarget)!\pv{height}!270:(\tikztotarget)$)
    .. (\tikztotarget)\tikztonodes}},
    settings/.code={\tikzset{quiver/.cd,#1}
        \def\pv##1{\pgfkeysvalueof{/tikz/quiver/##1}}},
    quiver/.cd,pos/.initial=0.35,height/.initial=0}

\tikzset{tail reversed/.code={\pgfsetarrowsstart{tikzcd to}}}
\tikzset{2tail/.code={\pgfsetarrowsstart{Implies[reversed]}}}
\tikzset{2tail reversed/.code={\pgfsetarrowsstart{Implies}}}
\tikzset{no body/.style={/tikz/dash pattern=on 0 off 1mm}}

\usepackage[mathscr]{eucal}
\usepackage{hyperref}

\newcommand\C{\mathscr C}

\newcommand\id{\mathrm{id}}
\newcommand\Poset{\mathbf{Poset}}
\newcommand\Grp{\mathbf{Grp}}
\newcommand\Top{\mathbf{Top}}
\newcommand\Vect{\mathbf{Vect}}
\newcommand\Rel{\mathbf{Rel}}

\newcommand\CA{\mathrm{CA}}

\newcommand\GCA{\mathrm{GCA}}
\newcommand\Set{\mathbf{Set}}
\newcommand\End{\mathrm{End}}
\newcommand\Aut{\mathrm{Aut}}

\newcommand\Res{\mathrm{Res}}
\newcommand\Hom{\mathrm{Hom}}
\newcommand\Obj{\mathrm{Obj}}

\newcommand\Sym{\mathrm{Sym}}

\newtheorem{theorem}{Theorem}[section]
\newtheorem{definition}[theorem]{Definition}
\newtheorem{corollary}[theorem]{Corollary}
\newtheorem{proposition}[theorem]{Proposition}
\newtheorem{lemma}[theorem]{Lemma}
\newtheorem{example}[theorem]{Example}

\title{A categorical framework for cellular automata}
\author{A. Castillo-Ramirez\footnote{Email: alonso.castillor@academicos.udg.mx}, A. Vazquez-Aceves, and A. Zaldivar-Corichi \\ Centro Universitario de Ciencias Exactas e Ingenier\'ias, Universidad de Guadalajara, M\'exico.}
\date{}

\begin{document}

\maketitle

\begin{abstract}
This paper proposes a generalized framework for cellular automata using the language of category theory, extending the classical definition beyond set-theoretic constraints. For an arbitrary category $\mathscr{C}$ with products, we define $\mathscr{C}$-cellular automata as morphisms $\tau : A^G \to B^G$ in $\mathscr{C}$, where the alphabets $A$ and $B$ are objects in $\mathscr{C}$ and the universe is a group $G$. We show that $\mathscr{C}$-cellular automata form a subcategory of $\mathscr{C}$ closed under finite products, and that they satisfy a categorical version of the Curtis-Hedlund-Lyndon theorem. For two arbitrary group universes $G$ and $H$, we extend our theory to define generalized $\C$-cellular automata as morphisms $\tau : A^G \to B^H$ constructed via a group homomorphism $\phi : H \to G$. Finally, we prove that generalized $\mathscr{C}$-cellular automata form a subcategory of $\mathscr{C}$ with a finite weak product involving the free product of the underlying group universes. This framework unifies existing concepts and provides purely categorical proofs of foundational results in the theory of cellular automata. \\

\textbf{keywords}: Cellular automata; Generalized cellular automata; Category theory; Categorical product; Curtis-Hedlund-Lyndon theorem.
\end{abstract}

\section{Introduction}

Category theory is a powerful foundation of mathematics that provides a unifying framework to describe and analyze different structures and their relationships. Many constructions and concepts that appear in different contexts throughout mathematics, such as quotient spaces, direct products, adjoints, and duality, are unified with a precise definition in category theory. In particular, a \emph{categorical product}, or simply a \emph{product}, in a category is a generalization of the Cartesian product of sets, the direct product of groups, and the product topology of topological spaces. This captures the idea of combining a family of objects to produce a new object, while preserving in some way its relationship with the original objects.  

Cellular automata have traditionally been studied as transformations defined over a discrete group $G$ (the universe) with states drawn from a finite set $A$ (the alphabet). While the classical theory, often formalized using symbolic dynamics, group theory and topology, is well-developed (e.g., see \cite{Cec2010}) there has been a growing interest in understanding these systems through more abstract mathematical structures. This paper proposes a generalized framework for cellular automata using the language of category theory, moving beyond the constraints of set-theoretic definitions. 

The main goal of this paper is to generalize the traditional definition of cellular automata by replacing the alphabet set with an object $A$ in an arbitrary category $\C$ with products. In this framework, we introduce the notion of $\C$-cellular automata as a morphism $\tau : A^G \to B^G$, where $G$ is a group, defined via a local defining morphism. This definition allows us for a systematic study of $\C$-cellular automata when $\C$ is a concrete category such as $\Set$, $\Grp$, $\Top$, $\Vect_\mathbb{K}$ and $\Poset$, but also when $\C$ is an abstract category such as $\Poset(P)$, where $P$ is a complete lattice, and $\Rel$. We prove that $\C$-cellular automata form a subcategory of $\C$, denoted as $\CA_{\C}(G)$ that satisfy properties analogous to its classical counterpart, the category of $\Set$-cellular automata. A central result of this paper is the proof of a categorical version of the Curtis-Hedlund-Lyndon theorem, which characterizes $\C$-cellular automata as morphisms that are both $G$-equivariant and uniform; this shows that the essence of this well-known theorem is categorical, not topological as usually assumed. We also prove that the category $\CA_{\C}(G)$ itself is closed under finite products, which is a generalization of the work presented in \cite{CA-prod}.   

Furthermore, we extend our investigation to define generalized $\C$-cellular automata as morphisms $\tau : A^G \to B^H$ constructed via a group homomorphism $\phi : H \to G$, as introduced in \cite{Vaz2022} for $\C = \Set$. We show that these also form a subcategory $\GCA_\C$ of $\C$, and that any pair of its objects $A^G$ and $A^H$ has a weak product $(A \times B)^{G \ast H}$, where $A \times B$ is the product in $\C$ and $G \ast H$ is the free product of groups. 

While various works have linked cellular automata to category theory \cite{Siltar,Cec2013,Ville2014,Kan1,Kan2}, they have restricted their scope to \emph{concrete categories}, which are essentially structured sets with structure-preserving functions. In contrast, we do not assume that our categories are concrete; consequently, we employ purely categorical methods (such as universal properties and commutative diagrams) to formulate all our definitions and proofs. 

The structure of this paper is as follows. In Section 2 we review some basic notions of category theory, including the product and coproduct, as well as some examples. In Section 3, we define the category $\CA_{\C}(G)$ of $\C$-cellular automata over a group $G$, we prove a categorical version of the Curtis-Hedlund-Lyndon theorem, and the we show that $\CA_{\C}(G)$ is closed under finite products. Finally, in Section 4 we introduce the category $\GCA_\C$ of generalized $\C$-cellular automata and show that it has a finite weak product. 

This work is an extended version of \cite{CA-prod} that significantly expands its theoretical scope, as all the definitions and proofs in \cite{CA-prod} were done for the particular case $\C = \Set$. Moreover, the proof of the categorical Curtis-Hedlund-Lyndon theorem is completely new.


\section{Categorical concepts}

In this section we shall review some basic notions of category theory; for a broader treatment see the standard textbooks \cite{Handbook,RomanCat,Sau1970,IntroCat}.

A \emph{category} $\C$ consists of a collection of \emph{objects}, which we denote by $\Obj(\C)$, and a collection of \emph{morphisms}, or \emph{arrows}, between objects. Each morphism $f$ has a domain $A \in \Obj(\C)$ and a codomain $B \in \Obj(\C)$; all the information that a morphism carries is condensed in the notation $f : A \to B$. Part of the essence of category theory is that objects are treated as if they are not necessarily sets, and morphisms are treated as if they are not necessarily functions. 

We denote the collection of morphisms with domain $A$ and codomain $B$ by $\C(A,B)$. Moreover, a category has a rule for composing morphisms: this means that to each pair of morphisms $f : A \to B$ and $g : B \to C$ there is a morphism $g\circ f: A \to C$, called the \emph{composition} of $g$ with $f$. The composition of morphisms satisfies the following two axioms: 
\begin{enumerate}
\item For any three morphisms $f:A\to B$, $g:B\to C$ and $h:C\to D$, we have:
\[ h\circ(g\circ f) = (h\circ g)\circ f. \]

\item For every object $A\in\Obj(\C)$, there exists a morphism $\id_A\in\C( A,A)$ such that for any $B\in\Obj(\C)$ and morphisms $g\in\C(A,B)$ and $h\in\C(B,A)$ we have:
\[ g\circ \id_A=g\text{ and }\id_A\circ h=h. \]
\end{enumerate}

\begin{example}\label{Ex1}
We introduce some examples of categories.
\begin{enumerate}
\item The class of all sets together with the class of all functions between sets form a category denoted by $\Set$.

\item The class of all groups together with the class of all group homomorphisms form a category denoted by  $\Grp$.

\item The class of all topological spaces together with the class of all continuous functions form a category denoted by $\Top$.

\item Given a field $\mathbb{K}$, the class of vector spaces over $\mathbb{K}$ together with the class of linear transformations form a category denoted by $\Vect_\mathbb{K}$.

\item The class of all partially ordered sets (posets) together with the class of all order-preserving functions between posets form a category denoted by $\Poset$.

\item Let $(P,\leq)$ be a poset. We can view $P$ as a category, where $\Obj(P) =P$, and given $a,b\in P$, there is a unique morphism from $a$ to $b$ if and only if $a\leq b$. By the properties of the order relation it is clear that all the category axioms are satisfied. We denote this category by $\Poset(P)$. 

\item The class of all sets together with all the class of all binary relations between sets form a category denoted by $\Rel$. 
\end{enumerate}
\end{example}

In the first five examples, $\Set$, $\Grp$, $\Top$, $\Vect_\mathbb{K}$, and $\Poset$, the objects of the categories are sets with additional structure, and the morphisms are functions that preserve these structures. These categories are called \emph{concrete categories}. On the other hand, the last two examples, $\Poset(P)$ and $\Rel$ are not concrete categories, they are also called \emph{abstract categories}, because morphisms are not functions between sets.

Saying that a diagram of morphisms like the following
    \[\begin{tikzcd}
        A \arrow[d, "f"'] \arrow[r, "g"] & B \arrow[ld, "h"] \\
        C                                &                  
        \end{tikzcd}\]
commutes is equivalent to say that $h \circ g = f$. In general, saying that a diagram of morphisms commutes is equivalent to saying that the composition of any two paths with the same domains and codomains are equal. Thus, a commutative diagram expresses various equations of composition of morphisms. 

An \emph{isomorphism} in $\C$ is a morphism $f : A \to B$ in $\C$ such that there exists a morphism $g : B \to C$ in $\C$, called the \emph{inverse} of $f$, satisfying 
\[ g\circ f=\id_A \text{ and } f\circ g=\id_B. \]
In this case, we say that $A$ and $B$ are isomorphic objects on $\C$ and we write $A\cong B$. It is straightforward to prove that the inverse of an isomorphism is unique and an isomorphism, and that the composition of isomorphisms is an isomorphism. If the collection of endomorphisms $\End(A):= \C(A,A)$ is a set (not a proper class) for some $A \in \C$, then $\End(A)$ is a monoid with the composition of morphisms, while the set $\Aut(A):=\{f\in\End(A) \ |\ f \text{ is an isomorphism}\}$ is its group of units. 

A \emph{covariant functor} $F$ from a category $\C$ to a category $\mathscr{D}$ assigns to each object $A \in \C$ an object $F(A) \in \mathscr{D}$, and to each morphism $f : A \to B$ in $\C$ a morphism $F(f) : F(A) \to F(B)$ in $\mathscr{D}$ such that 
\[ F(\id_A) = \id_{F(A)} \quad \text{and} \quad F(g \circ f) = F(g) \circ F(f),  \] 
for every morphism $g : B \to C$ in $\C$. The definition of a \emph{contravariant functor} $G$ from $\C$ to $\mathscr{D}$ is analogous, except that it assigns to each morphism $f : A \to B$ in $\C$ a morphism $G(f) : G(B) \to G(A)$ in $\mathscr{D}$ such that 
\[ G(\id_A) = \id_{G(A)} \quad \text{and} \quad G(g \circ f) = G(f) \circ G(g),  \]  
for every morphism $g : B \to C$ in $\C$.

\begin{definition}[product]\label{def-prod}
Let $\C$ be a category, $I$ a set of indices, and $\{A_{i}\}_{i\in I}$ a family of objects of $\C$. The \emph{product} of the family $\{A_{i}\}_{i\in I}$ is an object $P\in \C$ together with a family of morphisms $\{ \pi _{i}:P\to A_{i}\} _{i\in I}$ satisfying the following universal property: for any object $Q\in \C$ and any family of morphisms $\{\rho _{i}:Q\to A_{i}\}_{i\in I}$ there exists a unique morphism $f:Q\to P$ such that 
\begin{equation*} 
\rho _{i}=\pi _{i}\circ f\text{ for all }i\in I. 
\end{equation*} 
Each morphism $\pi _{i} : P \to A_i$ is called the $i$th \emph{projection}.
  \end{definition}

It turns out that if a product of a family exists, then it is unique up to isomorphism. Hence, we write $\prod\limits_{i\in I}A_{i}$ to denote the product of the family $\{ A_{i}\} _{i\in I}$. When we have a finite family of objects $\{ A_{i}\}_{i=1}^{n}$, we can also denote their product as $A_{1}\times A_{2}\times \dots \times A_{n}$.

In particular, the product of two objects $A$ and $B$ in $\C$ is an object $A\times B\in\C$ together with morphisms $\pi_A : A\times B\to A$ and $\pi_B : A\times B\to B$ satisfying the following universal property: for any object $X \in \C$ and any pair of morphisms $f:X\to A$ and $g:X\to B$ in $\C$, there exists a unique morphism $h:X\to A\times B$ such that the following diagram commutes:
\[ \begin{tikzcd}
& X \arrow[ld, "f"'] \arrow[rd, "g"] \arrow[d, "h"] & \\
A & A\times B \arrow[l, "\pi_A"] \arrow[r, "\pi_B"'] & B
\end{tikzcd} \]

We say that the category $\C$ itself has \emph{products} if any family of objects of $\C$ has a product, and we say that $\C$ has \emph{finite products} if any finite family of objects in $\C$ has a product (or, equivalently, if any pair of objects in $\C$ has a product). A \emph{weak product} of two objects $A\times B\in\C$ is defined analogously as the product in Definition \ref{def-prod}, except that the morphism $h:X\to A\times B$ is not required to be unique (see \cite[Sec. X.2]{Sau1970}).

\begin{example}\label{Examples of products}\text{\\}
We introduce a few examples of products in categories.
\begin{enumerate}
\item In the concrete categories described above, the product of a family of sets $\{A_i\}_{i\in I}$ is the Cartesian product of sets $\prod_{i \in I} A_i$ equipped with the corresponding structure so that the projections $\{\pi_j : \prod_{i \in I} A_i \to A_j \}_{j \in I}$ are morphisms in the corresponding categories (e.g., $\Set$, $\Grp$, $\Top$,  $\Vect_{\mathbb{K}}$, $\Poset$). However, there exist concrete categories where the product is not the Cartesian product together with the projections described above. When it does coincide, we say that the concrete category has \emph{concrete products} (see \cite[Ex. 10.55]{concrete}).

\item In the category $\Poset(P)$, where $P$ is a fixed poset, the product of two objects $a,b \in P$ is their least upper bound.

\item In the category $\Rel$, the product of a family of sets $\{A_i\}_{i\in I}$ is the disjoint union $\sqcup_{i \in I} A_i$ together with the dual of the canonical embeddings. 
\end{enumerate}
\end{example}

The following result is well-known, but we add its proof here for completeness, as it is often used in the rest of this paper. 

\begin{lemma}\label{Lemma, Universal Product Property}
Let $\C$ be a category, let $\{ A_{i}\} _{i\in I}$ a family of objects with product $\prod\limits_{i\in I}A_{i}$, and consider two morphisms 
\[ X\overset{f}{\underset{g}{\rightrightarrows}}\prod\limits_{i\in I}A_{i}. \]
Then, $f=g$ if and only if $\pi_i\circ f=\pi_i\circ g$ for all $i\in I$.
\end{lemma}

\begin{proof}
The first implication is trivial, so let us assume that $\pi_i\circ f=\pi_i\circ g$ for all $i\in I$. By the universal property of the product applied to the family of morphisms $\{ \rho_i:=\pi_i\circ f=\pi_i\circ g \}_{i \in I}$, there exists a unique morphism $h:X\to \prod\limits_{i\in I}A_{i}$ such that $\rho _{i}=\pi _{i}\circ h$ for all $i \in I$. Therefore, by the construction of $\rho_i$ and the uniqueness of $h$, we have that $f=h=g$.
\end{proof}

The dual notion of a product in a category, called a \emph{coproduct}, is defined by reversing the direction of the morphisms in the definition of a product. Despite this simple change in the definition, products and coproducts may behave dramatically different in a category. We will only need the definition of coproduct of two objects.

\begin{definition}[coproduct]\label{def-cop}
Let $\C$ be a category, and $A,B\in\C$ be objects of the category. A \emph{coproduct} of objects $A$ and $B$ in $\C$ is an object $A+ B\in\C$ together with morphisms $\iota_A:A\to A+ B$ and $\iota_B:B\to A+ B$ satisfying the following universal property: for any object $X$ and morphisms $f:A\to X$ and $g:B\to X$ there exists a unique morphism $h:X\to A+ B$ such that $h\circ\iota_A=f$ and $h\circ\iota_B=g$; this is equivalent of saying that the following diagram commutes:
\[ \begin{tikzcd}
& X & \\
A \arrow[ru, "f"] \arrow[r, "\iota_A"'] & A+ B \arrow[u, "h"] & B \arrow[l, "\iota_B"] \arrow[lu, "g"']
\end{tikzcd}\]
\end{definition}

\begin{example}\label{Ex-cop}
We introduce some examples of coproducts in categories. 
\begin{enumerate}
\item In the category of sets $\Set$, the coproduct of $A$ and $B$ corresponds to the disjoint union $A \sqcup B$ together with the canonical inclusions $\iota_A : A\to A \sqcup B$ and $\iota_B : B\to A \sqcup B$.

\item In the category of groups $\Grp$, the coproduct of $G$ and $H$ corresponds to the \emph{free product} of groups $G\ast H$ together with the canonical embeddings $\iota_G :G\to G\ast H$ and $\iota_H : H\to G\ast H$ (see \cite[Ex. 88]{RomanCat} or \cite[Sec. III.3]{Sau1970}). Explicitly, if $G$ and $H$ are given by the following presentations
\[ G= \langle S_G \; | \;  R_G \rangle \quad \text{ and } \quad H = \langle S_H \; | \; R_H \rangle, \]  
then the free product $G \ast H$ has presentation 
\[   G\ast H= \langle S_G + S_H \; | \; R_G + R_H \rangle, \]
where $+$ denotes the coproduct (disjoint union) in the category of sets (see \cite[Theorem 11.53]{Rotman}). It is then clear that the canonical embeddings $\iota_G:G\to G\ast H$ and $\iota_H:H\to G\ast H$ are defined by mapping the generators $S_G$ and $S_H$ to $S_G + S_H$, respectively. 

\item In the category $\Poset(P)$, where $P$ is a fixed poset, the product of two objects $a,b \in P$ is their greatest lower bound.
\end{enumerate}
\end{example}


\section{Cellular automata in categories with products}

\subsection{Configuration objects}

For the rest of the paper, let $\C$ be a category with products. Examples of these the concrete categories $\Set$, $\Grp$, $\Top$, $\Vect_{\mathbb{K}}$, $\Poset$, and the abstract categories $\Poset(P)$, when $P$ is a \emph{complete lattice} (i.e., any arbitrary family of elements in $P$ has a least upper bound), and $\Rel$. 

 The first step in order to provide a categorical framework for cellular automata is to examine powers of an object $A \in \C$ with itself. Given a set $I$ and an object $A \in \C$, we denote the product of $I$ copies of $A$ as
\[    A^{I}:=\prod\limits_{i\in I}A.  \]
  For any subset $S \subseteq I$, the universal property of the product $A^S$ implies that there is a unique morphism 
\[   \Res_S^I : A^{I}\to A^{S} \]
such that
  \begin{equation} \label{eq-res}
  \pi _{s}^I=\pi _{s}^{S}\circ\Res_S^I \ \  \forall s\in S,
   \end{equation}
       where $\pi _s^{S}$ is a projection of $A^{S}$ and $\pi _s^I$ is a projection of $A^{I}$. We call $\Res_S^I : A^{I}\to A^{S}$ the \emph{restriction morphism}. 
  
    \begin{lemma}\label{Lemma, transitivity of res}
   For any object $A\in\C$ and sets $S\subseteq T\subseteq I$, then 
         \[ \Res_S^I = \Res_S^T \circ\Res_T^I .\]
    \end{lemma}
    \begin{proof}
      By the universal property of the product, the following diagram commutes for every $s \in S$:
           \[\begin{tikzcd}[ampersand replacement=\&]
        {A^I} \& {A^T} \& {A^S} \\
        \& A
        \arrow["{\Res_S^T}", from=1-2, to=1-3]
        \arrow["{\Res_T^I}", from=1-1, to=1-2]
        \arrow["{\pi_s^I}"', from=1-1, to=2-2]
        \arrow["{\pi_s^T}"', from=1-2, to=2-2]
        \arrow["{\pi_s^S}", from=1-3, to=2-2]
        \end{tikzcd}\]
Since $\Res_S^I : A^I\to A^S$ is the unique morphism such that $\pi_s^I=\pi_s^S \circ \Res_S^I$ for all $s \in S$, we must have that  $\Res_S^I = \Res_S^T \circ\Res_T^I$.
\end{proof}

    Consider a function $f:I\to J$ between sets $I$ and $J$, and fix an object $A\in\C$. For each $i\in I$, we have the morphism $\pi_{f(i)}^J : A^J\to A$, and, by the universal property of the product, there exists a unique morphism
\[      f^*_A : A^J \to A^I \]
such that:
   \begin{equation}\label{eq-star}
    \pi^J_{f(i)}=\pi_i^I \circ f^*_A, \quad \forall i\in I. 
     \end{equation}
We call the morphism $f_A^* : A^J \to A^I$ the \emph{pullback}\footnote{The use of this term is inspired by the pullback in differential geometry, not the pullback in category theory.} of the function $f : I \to J$.

    \begin{proposition}\label{le-comp-pullback}
      For any functions $I\xrightarrow{f}J\xrightarrow{g}K\in\Set$ and any object $A\in\C$, we have
        $$(g\circ f)^*_A=f^*_A\circ g^*_A$$
    \end{proposition}
    \begin{proof}
 The definition of the pullbacks $f^*_A$ and $g^*_A$ imply that the following diagram commutes: 
      \[\begin{tikzcd}
            {A^K} & {A^J} & {A^I} \\
            & A
            \arrow["{g^*_A}", from=1-1, to=1-2]
            \arrow["{\pi^K_{g(f(i))}}"', from=1-1, to=2-2]
            \arrow["{f^*_A}", from=1-2, to=1-3]
            \arrow["{\pi^J_{f(i)}}"', from=1-2, to=2-2]
            \arrow["{\pi^I_i}", from=1-3, to=2-2]
        \end{tikzcd}\]
        for all $i\in I$. By the definition of $(g\circ f)^*_A$ we have 
        \[ \pi^I_i\circ (g\circ f)^*_A=\pi^K_{g(f(i))}=\pi_i^I \circ (f^*_A\circ g^*_A), \quad \forall i \in I,  \] 
        so the result follows by the uniqueness of the pullback. 
    \end{proof}

   It is clear that for any $I \in \Set$ and $A \in \C$, $(id_{I})^*_A=id_{A^I}$; therefore, the previous proposition implies that the pullback $(\_)_A^*$ defines a contravariant functor from $\Set$ to $\C$. 
   
   It is also clear that if $\alpha :I\to I$ is a permutation of the set $I$, i.e. $\alpha \in \Sym(I)$, then $\alpha^*_A\in \Aut(A^{I})$, for any object $A \in \C$. To simplify the notation, we will generally write $\alpha^*$ instead of $\alpha^*_A$ when the object $A \in \C$ is clear from the context. 

  We shall focus now on the case when the set of indices for the product of $A \in \C$ with itself is given by a group $G \in \Grp$. Recall that for each $g\in G$, the right translation $R_{g}:G\rightarrow G$, given by $R_g(x):= xg$, $\forall x \in G$, is a permutation of $G$. Hence, the pullback $R_g^* : A^G \to A^G$ is an automorphism in $\C$ that satisfies 
    \[ R_g^* \circ R_h^* = (R_{h} \circ R_{g})^* = R_{gh}^*, \quad \forall g,h \in G.  \]
Define the \emph{shift action} of $G$ on $A^G$ as the assignment $\varphi :G\rightarrow \Aut(A^{G})$ given by 
\[     \varphi (g) = R_{g}^* \quad \forall g\in G. \]
When $\C=\Set$, this construction coincides with the usual shift action of $G$ on $A^G$ (see \cite[Sec 1.1]{Cec2010}).

For any $g \in G$, we write the image of the shift action of $G$ on $A^G$ by $\varphi_g := \varphi (g) \in \Aut(A^{G}) $. 

\begin{lemma}\label{le-shift-action}
For each $g \in G$, the automorphism $\varphi _{g}:A^{G}\rightarrow A^{G}$ is the only morphism of $\C$ such that
\[   \pi^G _{hg} = \pi^G_{h}\circ \varphi _{g}, \quad \forall g,h\in G, \] 
where $\pi^G_h : A^G \to A$ is the projection morphism. Moreover,
 \[ \varphi_{gh} = \varphi_g \circ \varphi_h, \quad \forall g,h \in G.  \] 
    \end{lemma}
    \begin{proof}
    This follows by the definition (\ref{eq-star}) of the pullback $R_{g}^*$. 
    \end{proof}

   We say that a product $A^{G}\in \C$ together with the shift action of $G$ on $A^G$ is a \emph{configuration object} with alphabet $A\in \C$ and universe $G\in \Grp$. 

   We finish this section with a technical result that will be useful later. For any subset $S$ of $G$ and for any $g\in G$, consider the translation of $S$ given by $Sg := \{sg : s \in S \}$. By the universal property of the product, there exists a unique morphism $\overset{\to}{g}:A^{Sg}\to A^S$ such that the following diagram commutes for every $s \in S$: 
  \[\begin{tikzcd}[ampersand replacement=\&]
        A^{Sg} \&\& A^S \\
        \& A
        \arrow["{\pi^{Sg}_{sg}}"', from=1-1, to=2-2]
        \arrow["{\pi^S_s}", from=1-3, to=2-2]
        \arrow["\overset{\to}{g}", from=1-1, to=1-3]
    \end{tikzcd}\]

    \begin{lemma}\label{Le-trans-res}
    Let $A^G\in\C$ be a configuration object, let $S\subseteq G$, and let $g\in G$. Then
    \[ \overset{\to}{g}\circ\Res^G_{Sg}=\Res^G_S\circ\varphi_{g}. \]
    \end{lemma}
    \begin{proof}
By Lemma \ref{le-shift-action} and the definition (\ref{eq-res}) of the restriction morphism, we have that for all $s \in S$,
      \[ \pi^{Sg}_{sg} \circ \Res^G_{Sg}= \pi^G_{sg} = \pi^G_s \circ \varphi_{g}=\pi^S_s\circ\Res^G_S\circ\varphi_{g}. \]
        By its definition, $\pi^{Sg}_{sg} = \pi^S_s \circ \overset{\to}{g}$, so we have 
        \[  \pi^S_s\circ \overset{\to}{g}\circ\Res^G_{Sg} = \pi^S_s\circ\Res^G_S\circ \varphi_{g},  \]
        for all $s \in S$. The result follows by Lemma \ref{Lemma, Universal Product Property}.
    \end{proof}

 
 \subsection{Definition of $\C$-cellular automata}

We generalize Definition 1.4.1 in \cite{Cec2010} regarding cellular automata over a group universe $G$ by allowing the alphabet $A$ to be not only a set but an object in an arbitrary category $\C$ with products. 

    \begin{definition}[$\C$-cellular automata]\label{C-CA}
    Let $\C$ be a category with products and let $G$ be a group. Let $A^{G},B^{G}\in \C$ be two configuration objects. A morphism $\tau : A^{G} \to B^{G} \in \C$ is a \emph{$\C$-cellular automaton} over $G$ if there exist a finite subset $S\subseteq G$ and a morphism $\mu : A^{S} \to B \in \C$ such that
    \[    \pi _{g}^G \circ \tau =\mu \circ \Res^G_S \circ \varphi_{g} \quad \forall g\in G;
    \end{equation*}
equivalently, this means that the following diagram commutes for every $g \in G$:
     \[\begin{tikzcd}[ampersand replacement=\&]
        {A^G} \&\&\& {B^G} \\
        \& {A^{G}} \\
        \&\& {A^S} \\
        \&\&\& B
        \arrow["{\tau}", from=1-1, to=1-4]
        \arrow["{\pi^G_g}", from=1-4, to=4-4]
        \arrow["{\varphi_g}"', from=1-1, to=2-2]
        \arrow["{\Res^G_S}"', from=2-2, to=3-3]
        \arrow["\mu"', from=3-3, to=4-4]
    \end{tikzcd}\]
The set $S$ is called a \emph{neighborhood} of $\tau$ and the morphism $\mu : A^S \to B$ is called a \emph{local defining morphism} of $\tau$.
    \end{definition}

It is clear that a $\Set$-cellular automaton over $G$ coincides with a cellular automaton over $G$ as in \cite[Ch. 1]{Cec2010}. We shall prove that $\C$-cellular automata satisfy many analogous to $\Set$-cellular automata.

\begin{example}\label{ex-id}
For any configuration object $A^G \in \C$ and $g\in G$, Lemma \ref{le-shift-action} and the definition of restriction morphism (\ref{eq-res}) imply that
    \[ \pi^G_g\circ\id_{A^G}=\pi^G_g= \pi^G_e \circ \varphi_{g} = \pi^{\{e\}}_e \circ \Res^G_{\{e\}}\circ\varphi_{g}, \]
   where $e \in G$ is the group identity. This shows that $\id_{A^G} : A^G \to A^G$ is a $\C$-cellular automaton with local defining morphism $\pi^{\{e\}}_e : A^{\{e\}} \to A$. 
    \end{example}

    \begin{proposition}\label{Prop, C-equivariant automata}
    Every $\C$-cellular automaton $\tau :A^{G}\to B^{G}$ is \emph{$G$-equivariant} in the sense that
    \[    \tau\circ\varphi_g=\varphi_g\circ\tau \quad \forall g\in G. \]
    \end{proposition}
    \begin{proof}
    Let $\mu :A^{S}\to B$ be a local defining morphism of $\tau$. Observe that by Lemma \ref{le-shift-action}, for every $g,k \in G$, 
    \[    \pi^G _{g} \circ \varphi_{k} \circ \tau  =\pi^G_{gk}\circ \tau =\mu \circ\Res^G_S\circ \varphi _{gk} =\mu \circ \Res^G_S\circ \varphi _{g}\circ\varphi _{k}=\pi^G_{g}\circ \tau \circ \varphi _{k}. \]
    The result follows by Lemma \ref{Lemma, Universal Product Property}.
    \end{proof}
    \begin{corollary}\label{cor-eq}
   For $A,B \in \C$, let $\tau : A^G \to B^G$ and $\sigma : A^G \to B^G$ be two $G$-equivariant morphisms in $\C$. If $\pi_e^G \circ \tau = \pi_e^G \circ \sigma$, then $\tau = \sigma$. 
    \end{corollary}
   \begin{proof}
    By Lemma \ref{le-shift-action} and Proposition \ref{Prop, C-equivariant automata}, we have that, for every $g \in G$,
   \[   \pi^G_g \circ \tau = \pi^G_e \circ (\varphi_{g} \circ \tau) = (\pi^G_e \circ \tau) \circ \varphi_{g} = \pi^G_e \circ (\sigma \circ \varphi_{g}) = (\pi^G_e \circ \varphi_{g}) \circ \sigma=\pi^G_g \circ\sigma.  \]
        It follows by Lemma \ref{Lemma, Universal Product Property} that $\tau=\sigma$.
   \end{proof}

    \begin{theorem}\label{th-comp}
    Let $\tau:A^{G}\to B^{G}$ and $\sigma:B^{G}\to C^{G}$ be two $\C$-cellular automata over $G$ with neighbourhoods $T \subseteq G$ and $S \subseteq G$, respectively. Then, the composition $\sigma\circ\tau:A^{G}\to C^{G}$ is a $\C$-cellular automaton over $G$ with neighbourhood $TS := \{ts : t \in T, s \in S \}$. 
    \end{theorem}

    \begin{proof}
    Let $\mu:A^T \to B$ and $\nu: B^S \to C$ be local defining morphisms of $\tau$ and $\sigma$, respectively. By Proposition \ref{Prop, C-equivariant automata}, for any $k\in G$, we have
    \begin{equation}\label{eq-comp}
  \pi_k^G \circ (\sigma\circ\tau) = (\pi^G_k\circ\sigma)\circ\tau= \nu\circ\Res^G_S\circ (\varphi_{k}\circ\tau) =  \nu\circ\Res^G_S\circ\tau\circ\varphi_{k}. 
\end{equation}

For $s\in S$, consider the projection $\pi^S_S:B^S\to B$. Then, by the definition of restriction (\ref{eq-res}) and Lemma \ref{Le-trans-res}, 
   \[  \pi^S_s \circ (\Res^G_S \circ \tau) = (\pi^S_s \circ \Res^G_S) \circ \tau =  \pi^G_s \circ \tau = \mu\circ\Res^G_T\circ\varphi_{s} = \mu \circ \overset{\to}{s}\circ\Res^G_{Ts}. \]
  This means that the following diagram commutes: 
     
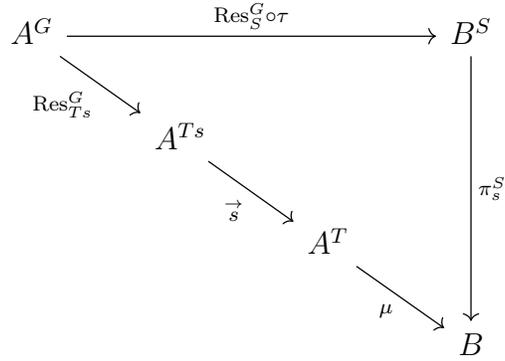
\[\begin{tikzcd}[ampersand replacement=\&]
        {A^G} \&\&\& {B^S} \\
        \& {A^{Ts}} \\
        \&\& {A^T} \\
        \&\&\& B
        \arrow["{\Res^G_S\circ\tau}", from=1-1, to=1-4]
        \arrow["{\pi^S_s}", from=1-4, to=4-4]
        \arrow["{\Res^G_{Ts}}"', from=1-1, to=2-2]
        \arrow["{\overset{\to}{s}}"', from=2-2, to=3-3]
        \arrow["\mu"', from=3-3, to=4-4]
    \end{tikzcd}\]
  On the other hand, by Lemma \ref{Lemma, transitivity of res}, we can factor the morphism $\Res_{Ts}$, and obtain the commutative diagram 
   
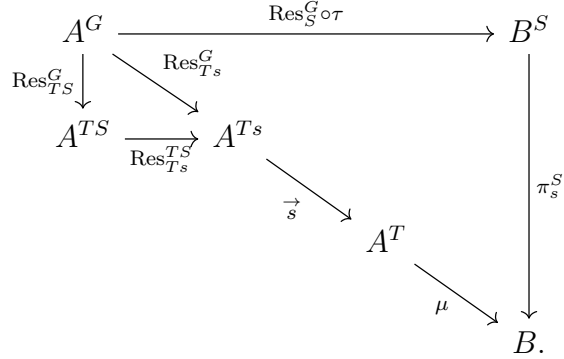
\[\begin{tikzcd}[ampersand replacement=\&]
        {A^G} \&\&\& {B^S} \\
        {A^{TS}} \& {A^{Ts}} \\
        \&\& {A^T} \\
        \&\&\& B.
        \arrow["{\Res^G_S\circ\tau}", from=1-1, to=1-4]
        \arrow["{\pi^S_s}", from=1-4, to=4-4]
        \arrow["{\Res^G_{Ts}}", from=1-1, to=2-2]
        \arrow["{\overset{\to}{s}}"', from=2-2, to=3-3]
        \arrow["\mu"', from=3-3, to=4-4]
        \arrow["{\Res^G_{TS}}"', from=1-1, to=2-1]
        \arrow["{\Res^{TS}_{Ts}}"', from=2-1, to=2-2]
    \end{tikzcd}\]

  By  the universal property of the product applied to the morphisms $\mu\circ\overset{\to}{s}\circ\Res^{TS}_{Ts} : A^{TS}\to B$ there exists a unique morphism $\Phi : A^{TS}\to B^S$ that makes the following diagram commute:
    \[\begin{tikzcd}
        {A^G} &&&& {B^S} \\
        & {A^{TS}} \\
        && {A^{Ts}} \\
        &&& {A^T} \\
        &&&& B
        \arrow["{\Res^G_S\circ\tau}", from=1-1, to=1-5]
        \arrow["{\Res^G_{TS}}"', from=1-1, to=2-2]
        \arrow["{\pi^S_s}", from=1-5, to=5-5]
        \arrow["\Phi"', from=2-2, to=1-5]
        \arrow["{\Res^{TS}_{Ts}}"', from=2-2, to=3-3]
        \arrow["{\overset{\to}{s}}"', from=3-3, to=4-4]
        \arrow["\mu"', from=4-4, to=5-5]
    \end{tikzcd}\]
    Hence, by Lemma \ref{Lemma, Universal Product Property}, we have 
     \[ \Phi\circ\Res^G_{TS}=\Res^G_S\circ\tau. \]
     Finally, by (\ref{eq-comp}),  
    \[       \pi_k\circ(\sigma\circ\tau)=\nu\circ (\Res^G_S\circ\tau ) \circ\varphi_{k} = \nu\circ\Phi\circ\Res^G_{TS}\circ\varphi_{k}, \quad \forall k \in G, \]
    which proves that $\sigma\circ\tau$ is a $\C$-cellular automaton over $G$ with local defining morphism $\nu\circ\Phi : A^{TS}\to C$.
    \end{proof}

    Theorem \ref{th-comp} and Example \ref{ex-id} show that the collection of configuration objects $A^G \in \C$ and $\C$-cellular automata form a \emph{subcategory} of $\C$ (see \cite[p. 12]{RomanCat}). We shall denote this category by $\CA_{\C}(G)$.

    \subsection{A categorical Curtis-Hedlund-Lyndon theorem}

The following definition captures, in a categorical way, the essence of a uniformly continuous map in the prodiscrete uniform structure of $A^G$ when $A \in \Set$ (see \cite[Sec. 1.9]{Cec2010}). 

    \begin{definition}[local and uniform morphism]
 Let $A,B\in\C$ be objects and let $I$ be a set. We say that a morphism $\mu : A^I \to B$ is \emph{local} if there exists a finite subset $S\subseteq I$ and a morphism $\mu':A^S \to B$ such that 
 \[ \mu=\mu'\circ\Res^I_S. \]
 We say that a morphism $f:A^I\to B^I$ is \emph{uniform} if, for every $i\in I$, the morphism $\pi^I_i\circ f : A^I \to B$ is local.
\end{definition} 

    \begin{lemma}\label{Le-unif}
        Every $\C$-cellular automaton is uniform.
    \end{lemma}

    \begin{proof}
        Let $\tau : A^G \to B^G$ be a $\C$-cellular automaton and let $\mu : A^S \to B$ be a local defining morphism of $\tau$. By Definition \ref{C-CA} and Lemma \ref{Le-trans-res}, for every $g \in G$, we have
      \[ \pi^G_g\circ\tau=\mu\circ\Res^G_S\circ\varphi_{g} = \mu\circ\overset{\to}{g}\circ\Res^G_{Sg}.  \]
      Since $Sg \subseteq G$ is a finite set, this shows that the morphism $\pi_g\circ\tau$ is local for every $g \in G$. Therefore, $\tau$ is uniform. 
    \end{proof}

The main result of this section is a categorical version of the well-known Curtis-Hedlund-Lyndon Theorem \cite[Theorem 1.8.1]{Cec2010}.

  \begin{theorem}
Let $A,B\in\C$ be objects and let $G$ be a group. A morphism $\tau : A^G \to B^G \in \C$ is a $\C$-cellular automaton if and only if $\tau$ is $G$-equivariant and uniform.
    \end{theorem}

    \begin{proof}
        If $\tau : A^G \to B^G \in \C$ is a $\C$-cellular automaton, then it is $G$-equivariant and uniform by Proposition \ref{Prop, C-equivariant automata} and Lemma \ref{Le-unif}).

        Conversely, suppose that $\tau : A^G \to B^G \in \C$ is a $G$-equivariant and uniform morphism. Since $\tau$ is uniform, for each $g\in G$, there exists a finite subset $S_g\subseteq G$ and a morphism $\mu_g:A^{S_g}\to B$ such that 
        \[ \pi^G_g\circ\tau = \mu_g\circ\Res^G_{S_g}. \]
       By Lemma \ref{le-shift-action} and the $G$-equivariance of $\tau$, we have that for every $g \in G$, 
           \[          \pi^G_g\circ\tau =  \pi^G_e\circ\varphi_{g} \circ\tau =  \pi^G_e\circ\tau\circ\varphi_{g} = \mu_e\circ\Res^G_{S_e}\circ\varphi_{g}, \]
           where $e \in G$ is the group identity. This proves that $\tau$ is a $\C$-cellular automaton with local defining morphism $\mu_e : A^{S_e} \to B$.
    \end{proof}
    

    \subsection{Finite products in the category $\CA_\C(G)$}

    In this section we show that the subcategory $\CA_\C(G)$ of $\C$ is closed under finite products.  
    
    Consider a set $I$ and a morphism $f : A \to B$ in $\C$. By the universal property of the product applied to the family of morphisms $\{ f\circ\pi^I_i :A^I\to B\}_{i\in I}$, there exists a unique morphism
    \[ f_*^I:A^I\to B^I \]
    such that the following diagram commute for every $i \in I$: 
  \begin{equation}\label{eq-power}
  \begin{tikzcd}
        {A^I} & {B^I} \\
        A & B
        \arrow["{f_*^I}", from=1-1, to=1-2]
        \arrow["{\pi^I_i}"', from=1-1, to=2-1]
        \arrow["{\pi^I_i}", from=1-2, to=2-2]
        \arrow["f"', from=2-1, to=2-2]
    \end{tikzcd}
    \end{equation}
where we abuse the notation by denoting both projections by $\pi^I_i$. We call $f_{*}^I : A^I \to B^I$ the \emph{pushforward} of the morphism $f : A \to B$.

    \begin{lemma}\label{le-comp-push}
        For any set $I$ and for any pair of morphisms $f : A \to B$ and $g : B \to C$ in $\C$, we have 
        \[ (g\circ f)_*^I=g_*^I\circ f_*^I \]
    \end{lemma}
    \begin{proof}
      By the construction (\ref{eq-power}) of the morphisms $f_*^I$ and $g_*^I$, the following diagram commutes for every $i \in I$: 
            \[\begin{tikzcd}
            {A^I} & {B^I} & {C^I} \\
            A & B & C
            \arrow["{{f_*^I}}", from=1-1, to=1-2]
            \arrow["{{\pi^I_i}}"', from=1-1, to=2-1]
            \arrow["{g_*^I}", from=1-2, to=1-3]
            \arrow["{{\pi^I_i}}", from=1-2, to=2-2]
            \arrow["{\pi^I_i}", from=1-3, to=2-3]
            \arrow["f"', from=2-1, to=2-2]
            \arrow["g", from=2-2, to=2-3]
        \end{tikzcd}\]
 Therefore,
 \[ \pi^I_i \circ (g_*^I\circ f_*^I) = (g\circ f)\circ\pi^I_i, \quad \forall i\in I. \]
 On the other hand, by the construction of $(g\circ f)_*^I$, we have that
 \[ \pi^I_i \circ (g\circ f)_*^I= (g\circ f) \circ \pi^I_i, \quad \forall i\in I. \]
 The result follows by Lemma \ref{Lemma, Universal Product Property}.
    \end{proof}

    \begin{lemma}\label{le-expG}
    For any group $G$ and any morphism $f : A \to B$ in $\C$, the pushforward $f_*^G : A^G \to B^G$ is a $\C$-cellular automaton over $G$.
    \end{lemma}

    \begin{proof}
      By the construction (\ref{eq-power}) of $f^G$,   
        \[ \pi^G_g \circ f_*^G = f \circ \pi^G_g, \quad \forall g \in G.  \]
       By Lemma \ref{le-shift-action} and the definition of restriction morphism (\ref{eq-res}), we have that for all $g \in G$, 
        \[ f \circ \pi^G_g = f \circ \pi^G_e \circ \varphi_{g} = f \circ \pi^{\{e\}}_e \circ\Res^G_{\{e\}}\circ\varphi_{g}.  \]
      This proves that $f_*^G : A^G \to B^G$ is a $\C$-cellular automaton with local defining morphism $f \circ \pi^{\{e\}}_e : A^{\{e\}}\to B$.
\end{proof}

Lemmas \ref{le-comp-push} and \ref{le-expG} imply that the pushforward $(\_)_*^G$ defines a covariant functor from $\C$ to $\CA_\C(G)$. 

For any two objects $A, B \in C$, denote by $\pi_A : A \times B \to A$ the projection morphism from the product $A \times B \in \C$ to $A$. By Lemma \ref{le-expG}, the pushforward $\pi_A^G:=(\pi_A)_*^G : (A\times B)^G\to A^G$ is a $\C$-cellular automaton over $G$. 

 \begin{theorem}\label{th-prod}
        Let $A,B\in\C$ be two objects and let $G$ be a group. The configuration object $(A\times B)^G \in\C$ together with the $\C$-cellular automata $\pi_A^G : (A\times B)^G\to A^G$ and $\pi_B^G:(A\times B)^G\to B^G$ is the product of $A^G$ with $B^G$ in the category $\CA_\C(G)$.
    \end{theorem}

    \begin{proof}
       We must show that $(A\times B)^G$ together with $\pi_A^G : (A\times B)^G\to A^G$ and $\pi_B^G:(A\times B)^G\to B^G$ satisfy the universal property of the product in the category $\CA_\C(G)$. In other words, we must show that for any configuration object $C^G$ and any pair of $\C$-cellular automata $\alpha : C^G \to A^G$ and $\beta : C^G \to B^G$ there exists a unique $\C$-cellular automaton $\tau : A^G \to (A \times B)^G$ such that the following diagram commutes: 
\begin{equation} \label{eq-product}
\begin{tikzcd}
            & {C^G} \\
            {A^G} & {(A\times B)^G} & {B^G}
            \arrow["\alpha"', from=1-2, to=2-1]
            \arrow["\tau"', from=1-2, to=2-2]
            \arrow["\beta", from=1-2, to=2-3]
            \arrow["{\pi_A^G}", from=2-2, to=2-1]
            \arrow["{\pi_B^G}"', from=2-2, to=2-3]
        \end{tikzcd}
        \end{equation}
       
     By Definition \ref{C-CA}, then there exist finite subsets $S,T\subseteq G$ and morphisms $\mu:C^S \to A$ and $\eta:C^T\to B$ such that 
     \[ \pi^G _{g} \circ \alpha =\mu \circ\Res^G_S\circ\varphi_{g} \text{  and  } \pi^G_{g} \circ \beta =\eta \circ\Res^G_T\circ \varphi_{g}, \]
     for all $g \in G$. Consider the finite subset $V := S\cup T \subseteq G$. By the universal property of the product in $\C$ applied to the morphisms $\mu \circ \Res^V_S : C^V \to A$ and $\eta \circ \Res^V_T : C^V \to B$,  there is a unique morphism $\nu:C^V \to (A\times B)$ in $\C$ such that the following diagram commutes: 
      \[\begin{tikzcd}
            & {C^V} \\
            {C^S} && {C^T} \\
            A & {A\times B} & B
            \arrow["{{\Res^V_S}}"', from=1-2, to=2-1]
            \arrow["{{\Res^V_T}}", from=1-2, to=2-3]
            \arrow["\nu"{description}, from=1-2, to=3-2]
            \arrow["\mu"', from=2-1, to=3-1]
            \arrow["\eta", from=2-3, to=3-3]
            \arrow["{\pi_A}", from=3-2, to=3-1]
            \arrow["{\pi_B}"', from=3-2, to=3-3]
        \end{tikzcd}\]
        
 Again, by the universal property of the product in $\C$ applied to the family of morphisms $\{ \nu \circ \Res^G_V \circ\varphi_{g} : C^G\to (A\times B) \}_{g \in G}$ there exists a unique morphism $\tau:C^G\to (A\times B)^G$ such that
 \[ \pi^G_g \circ \tau = \nu \circ \Res^G_V \circ\varphi_{g}, \quad \forall g \in G. \]
This implies that $\tau : C^G \to (A \times B)^G$ is a $\C$-cellular automaton over $G$. 

By Lemma \ref{Lemma, transitivity of res} and the definition of $\nu:C^V \to (A\times B)$, we have that, for all $g \in G$,
        \begin{eqnarray*}
            \pi^G_g \circ \alpha &=& \mu \circ\Res^G_S\circ\varphi_{g} \\
                                          &=& \mu \circ (\Res^G_S\circ\Res^S_V )\circ\varphi_{g} \\ 
                                          &=& (\pi_A\circ\nu) \circ\Res^S_V\circ\varphi_{g} \\
                                           &=& \pi_A \circ (\pi^G_g \circ \tau) \\ 
                                           &=& \pi^G_g\circ\pi_A^G\circ\tau,
        \end{eqnarray*}
        where the last equality follows by the definition of the pushforward morphism $\pi_A^G$ (\ref{eq-power}). By Lemma \ref{Lemma, Universal Product Property}, $\alpha =\pi_A^G\circ\tau$. Similarly, $\beta = \pi_B^G \circ \tau$, which proves that diagram (\ref{eq-product}) commutes.  

 Now, in order to show the uniqueness of $\tau$, suppose that $\sigma:C^G\to (A\times B)^G$ is a $\C$-cellular automaton over $G$ such that $\alpha=\pi_A^G\circ\sigma$ and $\beta=\pi_B^G\circ\sigma$. Observe that for $X \in \{A,B\}$, 
 \[ \pi_X \circ \pi^G_e \circ \tau =\pi^G_e \circ \pi_X^G \circ\tau = \pi^G_e \circ \pi_X^G \circ \sigma= \pi_X \circ \pi^G_e \circ\sigma.  \]
   Hence, by Lemma \ref{Lemma, Universal Product Property}, we have $\pi_e^G \circ \tau = \pi_e^G \circ \sigma$. By Corollary \ref{cor-eq}, we have that $\tau=\sigma$.
    \end{proof}

The product of configuration objects $A^G$ and $B^G$ in $\CA_\C(G)$ given in Theorem \ref{th-prod} is isomorphic to the product $A^G \times B^G$ in $\C$; as products may be defined as small limits, this is a particular case of the general categorical principle that small limits commute with small limits \cite[p. 231]{Sau1970}. This shows that $\CA_\C(G)$ is closed under finite products.


\section{Generalized $\C$-cellular automata}

The main idea in the definition of \emph{generalized cellular automata}, introduced in \cite{Vaz2022}, is to consider cellular automata between configuration spaces over different groups. In this section, we provide a categorical definition for generalized cellular automata. 

For any two groups $G$ and $H$, denote by $\Hom(H,G)$ the set of all group homomorphisms from $H$ to $G$. 

\begin{definition}
Let $A^{G},B^{H}\in \C$ be two configuration objects, and let $\tau:A^{G}\rightarrow B^{H}$ be a morphism in $\C$. We say that $\tau$ is a $\phi$-cellular automaton, with $\phi \in \Hom(H,G)$, if there exist a finite subset $S\subseteq G$, and a morphism $\mu :A^{S}\to B\in \C$ such that:
\[ \pi _{h}\circ \tau =\mu \circ\Res^G_S\circ\varphi_{\phi (h)}, \quad \forall h\in H. \]
\end{definition}

In general, we say that $\tau:A^{G}\rightarrow B^{H}$ is a \emph{generalized $\C$-cellular automaton} if it is a $\phi$-cellular automaton for some $\phi \in \Hom(H,G)$.	

The proof of the following result is analogous to the proof of Proposition \ref{Prop, C-equivariant automata}.

\begin{lemma}\label{Lemma, phi equivariance}
Every $\phi$-cellular automaton $\tau:A^{G}\rightarrow B^{H}$ is $\phi$-equivariant in the sense that
\[ \tau\circ\varphi_{\phi(h)}=\varphi_h\circ\tau, \quad \forall h\in H. \]
\end{lemma}

Furthermore, analogous to Theorem \ref{th-comp}, the composition of two generalized $\C$-cellular automata is a generalized $\C$-cellular automaton (c.f., \cite[Theorem 2]{Vaz2022}).

\begin{theorem}\label{Theory, ACG composition}
Let $\tau:A^{G}\to B^{H}$ be a $\phi$-cellular automaton with neighborhood $T \subseteq G$ and let $\sigma : B^{H}\to C^{K}$ be a $\psi$-cellular automaton with neighborhood $S \subseteq H$. Then, $\sigma \circ \tau : A^G \to C^K$ is a $\phi \circ \psi$-cellular automaton with neighborhood $T\phi(S) := \{ t\phi(s) : t \in T, s \in S\}$. 
\end{theorem}

The previous theorem implies that the class of configuration objects $A^G$ in $\C$, where $G$ is any group, together with the class of generalized $\C$-cellular automata, form a subcategory of $\C$ that we denote by $\GCA_\C$. It is clear that, for any group $G$, the category $\CA_\C(G)$ is a subcategory of $\GCA_\C$, since every $\C$-cellular automaton is an $\id$-cellular automaton.  

\begin{lemma}\label{ex-star}
For any group homomorphism $\phi:H\to G$ and any object $A \in \C$, the pullback $\phi_A^*:A^G\to A^H$, as defined in (\ref{eq-star}), is a $\phi$-cellular automaton.
 \end{lemma}
 \begin{proof}
By the definition of pullback (\ref{eq-star}), Lemma \ref{le-shift-action}, and the definition of restriction morphism \ref{eq-res}, we have that for all $h\in H$, 
\[ \pi^H_h\circ\phi_A^*= \pi^G_{\phi(h)} = \pi^G_e\circ\varphi_{\phi(h)}=\pi^{\{e\}}_e \circ\Res^G_{\{e\}}\circ\varphi_{\phi(h)}. \]
This shows that $\phi_A^*:A^G\to A^H$ is a $\phi$-cellular automaton. 
\end{proof}
 
 Lemmas \ref{ex-star} and \ref{le-comp-pullback} imply that, for any fixed object $A \in \C$, the pullback $(\_)_A^*$ defines a contravariant functor from $\Grp$ to $\GCA_\C$.

The following result shows that every $\phi$-cellular automaton may be factorized as a composition of a unique $\C$-cellular automaton with a pullback of $\phi$ (c.f. \cite[Lemma 4]{Further}).

\begin{lemma}\label{ACG Factorization} 
For any generalized $\C$-cellular automaton $\tau:A^G\to B^H$ in $\GCA_\C$, there exists a group homomorphism $\phi:H\to G$ and a unique $\C$-cellular automaton $\tau_G : A^G\to B^G$ in $\CA_\C(G)$ such that 
\[ \tau=\phi_B^*\circ\tau_G \]
\end{lemma} 

\begin{proof} 
Let $\mu:A^S\to B$ be a local defining morphism for $\tau : A^G \to B^H$. By the universal property of the product applied to the family of morphisms $\{ \mu\circ\Res^G_S\circ\varphi_{g}:A^G\to B\}_{g \in G}$, there exists a unique morphism $\tau_G : A^G \to B^G$ such that 
\[ \pi^G_g\circ\tau_G = \mu\circ\Res^G_S\circ\varphi_{g}, \quad \forall g \in G.   \]
It is clear that $\tau_G:A^G\to B^G$ is a $\C$-cellular automaton. 

By Lemma \ref{ex-star}, for each $h\in H$ we have 
\[ \pi^H_h \circ \phi_B^* \circ \tau_G =  \pi^G_{\phi(h)} \circ \tau_G= \mu \circ \Res^G_S \circ \varphi_{\phi(h)} =\pi^H_h \circ \tau. \]
By Lemma \ref{Lemma, Universal Product Property}, we obtain that $\tau=\phi_B^*\circ\tau_G$. 

To prove the uniqueness of $\tau_G$, let $\sigma : A^G \to B^G$ be a $\C$-cellular automaton such that $\tau=\phi_B^*\circ\sigma$. If $e_G \in G$ and $e_H \in H$ are the corresponding identity elements, then 
\[ \pi^G_{e_G} \circ \sigma = \pi^G_{\phi(e_H)}\circ\sigma = \pi^H_{e_H} \circ\phi_B^*\circ\sigma=\pi^H_{e_H}\circ\tau = \pi^H_{e_H} \circ\phi_B^*\circ\tau_G = \pi^G_{e_G} \circ\tau_G \]
By Corollary \ref{cor-eq}, $\sigma=\tau_G$. 
\end{proof}

In order to prove that the category $\GCA_\C$ has a weak product, we need the following technical result about the commutativity of pullbacks and pushforwards. 

    \begin{lemma}\label{TN estrella}
For any morphism $f : A \to B$ in $\C$ and any group homomorphism $\phi : H \to G$, the following square commutes:
        \[\begin{tikzcd}
            {A^G} & {A^H} \\
            {B^G} & {B^H}
            \arrow["{\phi_A^*}", from=1-1, to=1-2]
            \arrow["{f_*^G}"', from=1-1, to=2-1]
            \arrow["{f_*^H}", from=1-2, to=2-2]
            \arrow["{\phi_B^*}"', from=2-1, to=2-2]
        \end{tikzcd}\]
    \end{lemma}

    \begin{proof}
       For any $h \in H$, consider the following diagram:
        \[\begin{tikzcd}
            {A^G} && {A^H} \\
            & A \\
            \\
            {B^G} && {B^H} \\
            & B
            \arrow["{{\phi_A^*}}", from=1-1, to=1-3]
            \arrow["{\pi^G_{\phi(h)}}"', from=1-1, to=2-2]
            \arrow["{{f_*^G}}"', from=1-1, to=4-1]
            \arrow["{\pi^H_h}", from=1-3, to=2-2]
            \arrow["{{f_*^H}}", from=1-3, to=4-3]
            \arrow[""{name=0, anchor=center, inner sep=0}, "{{\phi_B^*}}", from=4-1, to=4-3]
            \arrow["{\pi^G_{\phi(h)}}"', from=4-1, to=5-2]
            \arrow["{\pi^H_h}", from=4-3, to=5-2]
            \arrow["f"', between={0}{0.8}, no head, from=2-2, to=0]
            \arrow[between={0.2}{1}, from=0, to=5-2]
        \end{tikzcd}\]

Observe that the upper and lower triangles of this diagram commute by the definition of pullback (\ref{eq-star}). Furthermore, the two squares with edge $f$ commute by the definition (\ref{eq-power}) of the power morphisms $f_*^G$ and $f_*^H$. Therefore, the above diagram commutes, so 
\[ \pi^H_h\circ (\phi_B^*\circ f_*^G) = \pi^H_h\circ (f_*^H\circ\phi_A^*), \quad \forall h \in H. \]
Finally, by Lemma \ref{Lemma, Universal Product Property}, $\phi_B^*\circ f_*^G= f_*^H\circ\phi_A^*$, as required.
    \end{proof}

For any two groups $G$ and $H$, consider the canonical embeddings $\iota_G :G\to G\ast H$ and $\iota_H : H\to G\ast H$ into the free product $G \ast H$, which is the coproduct in the category $\Grp$ (see Example \ref{Ex-cop}). For any $A, B \in \C$, consider the following generalized $\C$-cellular automata:
\begin{align*}
 \iota_{A}^G & :=\pi_A^G\circ (\iota_G)_{A \times B}^* : (A \times B)^{G \ast H} \to (A \times B)^G \to A^G    \\
 \iota_{B}^H & :=\pi_B^H\circ (\iota_H)_{A \times B}^* : (A \times B)^{G \ast H} \to (A \times B)^H \to B^H,
 \end{align*}
 where $\pi_A^G : (A \times B)^G \to A^G$ and $\pi_B^H : (A \times B)^H \to B^H$ are the $\C$-cellular automata projections defined in Theorem \ref{th-prod}.
    
    \begin{theorem}\label{th-weak}
    Let $G$ and $H$ be two groups and let $A$ and $B$ be two objects in $\C$. The configuration object $(A\times B)^{(G*H)}$ together with the generalized $\C$-cellular automata $\iota_{A}^G  : (A \times B)^{G \ast H} \to A^G$ and $\iota_{B}^H : (A \times B)^{G \ast H} \to B^H$ is a weak product of $A^G$ and $B^H$ in the category $\GCA_\C$.
    \end{theorem}

    \begin{proof}
      Let $\alpha : C^K \to A^G$ and $\beta : C^K \to B^H$ be two generalized $\C$-cellular automata. By Lemma \ref{ACG Factorization}, there exist group homomorphisms $\phi : G \to K$ and $\psi : H \to K$ and unique $\C$-cellular automata $\alpha_K : C^K \to A^K$ and $\beta_K : C^K \to B^K$ such that the following diagram commutes:
        \[\begin{tikzcd}
            {A^G} && {C^K} && {B^H} \\
            & {A^K} && {B^K}
            \arrow["\alpha"', from=1-3, to=1-1]
            \arrow["\beta", from=1-3, to=1-5]
            \arrow["{\alpha_K}", from=1-3, to=2-2]
            \arrow["{\beta_K}"', from=1-3, to=2-4]
            \arrow["{\phi_A^*}", from=2-2, to=1-1]
            \arrow["{\psi_B^*}"', from=2-4, to=1-5]
        \end{tikzcd}\] 
     By Theorem \ref{th-prod}, there exists a unique $\C$-cellular automaton $\tau_K : C^K \to (A \times B)^K$ such that the following diagram commutes: 
             \[\begin{tikzcd}
            {A^G} && {C^K} && {B^H} \\
            & {A^K} && {B^K} \\
            && {(A\times B)^K}
            \arrow["\alpha"', from=1-3, to=1-1]
            \arrow["\beta", from=1-3, to=1-5]
            \arrow["{\alpha_K}"', from=1-3, to=2-2]
            \arrow["{\beta_K}", from=1-3, to=2-4]
            \arrow["{\tau_K}", from=1-3, to=3-3]
            \arrow["{\phi_A^*}", from=2-2, to=1-1]
            \arrow["{\psi_B^*}"', from=2-4, to=1-5]
            \arrow["{\pi_A^K}", from=3-3, to=2-2]
            \arrow["{\pi_B^K}"', from=3-3, to=2-4]
        \end{tikzcd}\] 

Now, we apply the pullback contravariant functor $(\_)_{A \times B}^*$ to the group homomorphisms in universal property of the coproduct in the category $\Grp$ (see Definition \ref{def-cop}) to obtain the following commutative diagram: 
        \[\begin{tikzcd}
            & {(A\times B)^K} \\
            {(A\times B)^G} && {(A\times B)^H} \\
            & {(A\times B)^{(G*H)}}
            \arrow["{\phi_{A\times B}^*}"', from=1-2, to=2-1]
            \arrow["{\psi_{A\times B}^*}", from=1-2, to=2-3]
            \arrow["{\gamma_{A\times B}^*}", from=1-2, to=3-2]
            \arrow["{(\iota_G)_{A\times B}^*}", from=3-2, to=2-1]
            \arrow["{(\iota_H)_{A\times B}^*}"', from=3-2, to=2-3]
        \end{tikzcd}\]
where $\gamma : G \ast H \to K$ is the unique group homomorphism  satisfying the universal property of the coproduct applied to $\phi$ and $\psi$.

By gluing the two previous diagrams, we obtain the following: 
        \[\begin{tikzcd}
            &&& {C^K} \\
            \\
            {A^G} && {A^K} && {B^K} && {B^H} \\
            & {(A\times B)^G} && {(A\times B)^K} && {(A\times B)^H} \\
            \\
            &&& {(A\times B)^{(G*H)}}
            \arrow["\alpha"', from=1-4, to=3-1]
            \arrow["{\alpha_K}"', from=1-4, to=3-3]
            \arrow["{\beta_K}", from=1-4, to=3-5]
            \arrow["\beta", from=1-4, to=3-7]
            \arrow["{\tau_K}", from=1-4, to=4-4]
            \arrow["{\phi_A^*}"', from=3-3, to=3-1]
            \arrow["{\psi_B^*}", from=3-5, to=3-7]
            \arrow["{\pi_A^G}", from=4-2, to=3-1]
            \arrow["{\pi_A^K}"', from=4-4, to=3-3]
            \arrow["{\pi_B^K}", from=4-4, to=3-5]
            \arrow["{\phi_{A\times B}^*}", from=4-4, to=4-2]
            \arrow["{\psi_{A\times B}^*}"', from=4-4, to=4-6]
            \arrow["{\gamma_{A\times B}^*}", from=4-4, to=6-4]
            \arrow["{\pi_B^H}"', from=4-6, to=3-7]
            \arrow["{(\iota_G)_{A\times B}^*}", from=6-4, to=4-2]
            \arrow["{(\iota_H)_{A\times B}^*}"', from=6-4, to=4-6]
        \end{tikzcd}\] 
        
    By Lemma \ref{TN estrella}, the left and right middle squares commute, so the whole diagram commutes. The result follows. 

    \end{proof}

The problem with showing that the weak product given in Theorem \ref{th-weak} is unique, and hence a product, is that the factorization of generalized $\C$-cellular automata given by Lemma \ref{ACG Factorization} is not unique in general. We say that $\tau : A^G \to B^H$ in $\GCA_\C$ has a \emph{unique factorization} if $\tau = \phi^*_B \circ \tau_G = \psi_B^* \circ \tau_G$ implies $\phi = \psi$. It was shown in \cite[Corollary 1]{Further}, that if $G$ is a torsion-free abelian group, then every non-constant generalized $\Set$-cellular automaton has a unique factorization; this result was generalized in \cite{SaloBlog} by removing the hypothesis of $G$ being abelian. However, note that constant generalized $\Set$-cellular automata never have a unique factorization: if $\tau : A^G \to B^H$ is constant, then $\tau = \phi^*_B \circ \tau_G = \psi_B^* \circ \tau_G$ for every pair of group homomorphisms $\phi$ and $\psi$. This implies that the weak product given in Theorem \ref{th-weak} is not a product in $\GCA_{\Set}$, but it is an open question if a product actually exists in $\GCA_{\Set}$. 


\section*{Acknowledgments}

The second author was supported by a SECIHTI \emph{Becas nacionales para estudios de posgrado}. The third author was supported by a SECIHTI research grant No. CBF-2023-2024-2630.

\bibliographystyle{plain}

\begin{thebibliography}{}

\bibitem[{S (2011)}]{concrete} Ad\'mek, J., Herrlich, H., Strecker, G.E.: Abstract and concrete categories: the joy of cats. Online Edition, 2004. \url{http://katmat.math.uni-bremen.de/acc/}

\bibitem[{S (2011)}]{Handbook}
Borceux, F.: Handbook of Categorical Algebra 1: Basic Category Theory, Cambridge University Press, Cambridge, 1994.

\bibitem[{S (2011)}]{Siltar}
Capobianco, S., Uustalu, T.: A categorical outlook on cellular automata. Journées Automates Cellulaires, 88--99 (2010). \url{https://hal.science/hal-00542015v1}

\bibitem[{S (2011)}]{Vaz2022}
Castillo-Ramirez, A., Sanchez-Alvarez, M., Vazquez-Aceves, A., Zaldivar-Corichi, A.: A generalization of cellular automata over groups. Comm. Algebra \textbf{51}(7), 3114--3123 (2023). \url{https://doi.org/10.1080/00927872.2023.2177663}

\bibitem[{S (2011)}]{Further}
Castillo-Ramirez, A., de~los Santos Ba\~{n}os, L.: Further results on generalized cellular automata. Comm. Algebra \textbf{52}(6), 2475--2488 (2024). \url{ https://doi.org/10.1080/00927872.2023.2301538}

\bibitem[{S (2011)}]{CA-prod}
Castillo-Ramirez, A., Vazquez-Aceves, A., Zaldivar-Corichi, A. (2026). Categorical Products of Cellular Automata. In: Riva, S., Richard, A. (eds) Cellular Automata and Discrete Complex Systems. AUTOMATA 2025. Lecture Notes in Computer Science \textbf{15831}. Springer, Cham. \url{https://doi.org/10.1007/978-3-032-01570-9_5}

\bibitem[{S (2011)}]{Cec2010}
Ceccherini-Silberstein, T., Coornaert, M.: Cellular Automata and Groups. Springer Monographs in Mathematics. Springer Cham, Switzerland, 2nd edition, 2023.

\bibitem[{S (2011)}]{Cec2013}
Ceccherini-Silberstein, T., Coornaert, M.: Surjunctivity and reversibility of cellular automata over concrete categories. In: Picardello, M. (eds.) Trends in Harmonic Analysis. Springer INdAM Series, vol 3. Springer, Milano. \url{https://doi.org/10.1007/978-88-470-2853-1_6}

\bibitem[{S (2011)}]{Kan1}
Fernandez, A., Maignan, L., Spicher, A.: Cellular Automata and Kan Extensions. In: Castillo-Ramirez, A., Guillon, P., Perrot, K. (eds.) 27th IFIP WG 1.5 International Workshop on Cellular Automata and Discrete Complex Systems (AUTOMATA 2021). Open Access Series in Informatics (OASIcs) \textbf{90}, pp. 7:1--7:12, Schloss Dagstuhl – Leibniz-Zentrum für Informatik, 2021. \url{ https://doi.org/10.4230/OASIcs.AUTOMATA.2021.7}

\bibitem[{S (2011)}]{Kan2}
Maignan, L., Spicher, A.: Causal Graph Dynamics and Kan Extensions. In: Harmer, R., Kosiol, J. (eds.) Graph Transformation. ICGT 2024. Lecture Notes in Computer Science \textbf{14774}, Springer, Cham, 2024. \url{https://doi.org/10.1007/978-3-031-64285-2_5}

\bibitem[{S (2011)}]{Sau1970}
Mac~Lane, S.: Categories for the Working Mathematician. Graduate Texts in Mathematics. Springer, New York, 2nd ed., 2010.

\bibitem[{S (2011)}]{RomanCat}
Roman, S.: An Introduction to the Language of Category Theory. Compact Textbooks in Mathematics. Birkhäuser Cham, Switzerland, 2017.

\bibitem[{S (2011)}]{Ville2014}
Salo, V., Törmä, T.: Category theory of symbolic dynamics. Theoret. Comput. Sci. \textbf{567}, 21--45 (2015). \url{https://doi.org/10.1016/j.tcs.2014.10.023}

\bibitem[{S (2011)}]{SaloBlog}
Salo, V.: About last weeks (2/n): A group-theoretic take on generalized cellular automata, Shift Happens (blog), 2024. \url{https://symbolicdynamicsandotherthings.wordpress.com/2024/03/12/about-last-weeks-2-n-a-group-theoretic-take-on-generalized-cellular-automata/}, last accessed 2025/03/10.

\bibitem[{S (2011)}]{IntroCat}
Simmons, H.
\newblock \emph{An introduction to category theory}. 
\newblock Cambridge University Press, Cambridge, 2011.
	
\end{thebibliography}

\end{document}